%% file: js.tex
\documentclass[12pt]{colt2015}
\usepackage{hyperref} 
\usepackage{microtype}
\usepackage{graphicx}

\title[sketching]{Sketching, Embedding, and Dimensionality
Reduction for Information Spaces\thanks{This research was funded in part by the
  NSF under grants CCF-0953066, CCF-0953754, CCF-1320719 and BIGDATA-1251049 and
by a Google Faculty Research Award.}}

\author[Abdullah et al.]{\Name{Amirali Abdullah} \Email{amirali@cs.utah.edu} \\ \addr University of Utah \AND \Name{Ravi Kumar}
  \Email{tintin@google.com} \\ \addr Google, Inc. \AND \Name{Andrew McGregor}
  \Email{mcgregor@cs.umass.edu} \\ \addr University of Massachusetts--Amherst
  \AND \Name{Sergei Vassilvitskii} \Email{sergeiv@google.com} \\ \addr Google,
  Inc. \AND \Name{Suresh Venkatasubramanian} \Email{suresh@cs.utah.edu} \\ \addr
University of Utah}

\input{macros}

\begin{document}
\maketitle
\begin{abstract}

\emph{Information distances} like the Hellinger distance and the Jensen-Shannon
divergence have deep roots in information theory and machine learning. They are
used extensively in data analysis especially when the objects being compared  are
high dimensional empirical probability distributions built from data. However, we lack common tools needed to actually use information distances in applications efficiently
and at scale with any kind of provable guarantees. We can't sketch these
distances easily, or embed them in better behaved spaces, or even reduce the
dimensionality of the space while maintaining the probability structure of the data. 

In this paper, we build these tools for information distances---both for the Hellinger distance and Jensen--Shannon divergence, as well as related measures, like the $\chi^2$ divergence. We first show that they can be sketched efficiently (i.e. up to multiplicative error in sublinear space) in the \emph{aggregate} streaming model. This result is exponentially stronger than known upper bounds for sketching these distances in the \emph{strict turnstile} streaming model. Second, we show a finite dimensionality embedding result for the Jensen-Shannon and $\chi^2$ divergences that preserves pair wise distances. Finally we prove
a dimensionality reduction result for the Hellinger, Jensen--Shannon, and
$\chi^2$ divergences that preserves the information geometry of the
distributions (specifically, by retaining the simplex structure of the
space). While our second result above already implies that these divergences can
be explicitly embedded in Euclidean space, retaining the simplex structure is
important because it allows us to continue doing inference in the reduced
space. In essence, we preserve not just the distance structure but the
underlying geometry of the space. 

\end{abstract}

\input{intro}

\input{related}
\input{sketchjs}
\input{embedjs}
\input{embedtr}
\input{embedtrsecond}
\input{simptosimp}

\section{Conclusions}
\label{sec:conclusions}

The embedding and sketching results we show here complements the known impossibility
results for sketching information distances  in the strict turnstile model, thus providing a more
complete picture of how these distances can be estimated in a stream. The
dimensionality reduction result essentially says that as long as the information
distance admits a ``Euclidean-like'' patch somewhere in the simplex, it can be
mapped to a lower dimensional space. This latter result is a little surprising
because the Hellinger distance exhibits more $\ell_1$ like behavior at the
corners of the simplex. In fact, we conjecture that if we limit ourselves to mappings that are not
contractive, then it is likely that the Hellinger distance will \emph{not} admit
accurate dimensionality reduction. 
\newpage
\bibliography{js}

\newpage
\end{document}

%% file: macros.tex
\usepackage{amsmath,amssymb,amsfonts}
\usepackage{verbatim}
\usepackage{amssymb}
\usepackage{mathrsfs}
\usepackage{lineno}
\usepackage{cleveref}
\usepackage{xspace}
\usepackage{nicefrac}

\usepackage[final,inline,nomargin,index]{fixme}
\fxsetup{theme=colorsig,mode=multiuser,inlineface=\sffamily,envface=\itshape}
\FXRegisterAuthor{su}{asu}{Suresh}
\FXRegisterAuthor{se}{ase}{Sergei}
\FXRegisterAuthor{an}{aan}{Andrew}
\FXRegisterAuthor{ra}{ara}{Ravi}
\FXRegisterAuthor{am}{aam}{Amir}

\newtheorem*{theorem*}{Theorem}
\newtheorem{claim}{Claim}[section]

\numberwithin{equation}{section}

\newcommand{\js}{\ensuremath{{\text{JS}}}}
\newcommand{\eps}{\varepsilon}

\newcommand{\reals}{{\mathbb R}}
\newcommand{\he}{\ensuremath{\text{He}}}
\newcommand{\w}{\omega}
\newcommand{\vecv}{{\textbf v}}
\newcommand{\vecu}{{\textbf u}}
\newcommand{\veca}{{\textbf a}}
\newcommand*\diff{\mathop{}\!\mathrm{d}}
\newcommand{\var}{\mathop{\rm var}}
\newcommand{\roundup}[1]{\lceil #1 \rceil}

\DeclareMathOperator{\sech}{sech}

%% file: intro.tex
\section{Introduction}
\label{sec:intro}

The space of \emph{information distances} includes many distances that are used
extensively in data analysis. These include the well-known Bregman divergences, the $\alpha$-divergences, and the $f$-divergences. In this work we focus on a subclass of the $f$-divergences that admit embeddings into some (possibly infinite-dimensional) Hilbert space, with a specific emphasis on the JS divergence. These divergences are used in statistical tests and estimators~\citep{hellingerestimate}, as well as in image analysis~\citep{imageanalysis1},
computer vision~\citep{vision1,shapematch1}, and text
analysis~\citep{text1, corpus1}. They were introduced by~\citet{csiszinformation}, and, in the most general case, also include
measures such as the Hellinger, JS, and $\chi^2$ divergences (here we consider a symmetrized variant of the $\chi^2$ distance).

To work with the geometry of these divergences effectively at scale and in high dimensions, we need algorithmic tools that can provide provably high quality approximate representations of the geometry. 
The techniques of \emph{sketching}, \emph{embedding}, and \emph{dimensionality
 reduction} have evolved as ways of dealing with this problem. 
 
A \textbf{sketch} for a set of points with respect to a property $P$ is a function that maps the data to a small summary from which property $P$ can be evaluated, albeit with some approximation error.  Linear sketches are especially useful for estimating a derived property of a data stream in a fast and compact way.\footnote{Indeed \citet*{2014turnstile} show that any optimal one-pass streaming sketch algorithm in the turnstile model can be reduced to a linear sketch with logarithmic space overhead.} 
Complementing sketching, \textbf{embedding} techniques are one to one mappings that transform a collection of points lying in one space $X$ to another (presumably easier) space $Y$, while approximately preserving distances between points. \textbf{Dimensionality reduction} is a special kind of embedding which preserves the structure of the space, while reducing its dimension. These embedding techniques can be used in an almost ``plug-and-play'' fashion to speed up many algorithms in data analysis: for example for near neighbor search (and classification), clustering, and closest pair calculations.

Unfortunately, while these tools have been well developed for norms like $\ell_1$ and $\ell_2$, we lack such tools for information distances. This is not just a theoretical concern: information distances are semantically more suited to many tasks in machine learning, and building the appropriate algorithmic toolkit to manipulate them efficiently would expand greatly the places where they can be used. 

\subsection{Our contributions}
\label{sec:our-contributions}

\paragraph*{Sketching information divergences.}
\citet*{GIM} proved an impossibility result, showing that a large class of information divergences cannot be sketched in sublinear space, even if we allow for constant factor approximations. This result holds in the \emph{strict turnstile streaming  model}---a model in which coordinates of two points $x$, $y \subset \Delta_d$ are increased incrementally and we wish to maintain an estimate of the divergence between them.  They left open the question of whether these divergences can be sketched in the \emph{aggregate} streaming model, where each element of the stream gives the $i$th coordinate of $x$ or $y$ in its entirety, but the coordinates may appear in an arbitrary order. We answer this in the affirmative for two important information distances, namely, the Jensen--Shannon and $\chi^2$ divergences.  

\begin{theorem}
A set of points $P$ under the Jensen--Shannon(\js) or $\chi^2$ divergence can be deterministically embedded into $O(\frac{d^2}{\eps} \log \frac{d}{\eps})$ dimensions under $\ell^2_2$ with $\eps$ additive error.  The same space bound holds when sketching $\js$ or $\chi^2$ in the \emph{aggregate} stream model.
\end{theorem}
\begin{corollary}
Assuming polynomial precision, an AMS sketch for Euclidean distance can reduce the dimension to $O \left(\frac{1}{\eps^2}
\log \frac{1}{\eps}  \log d  \right)$ for a $(1+\eps)$ multiplicative approximation in the aggregate stream setting.
\end{corollary}

\begin{theorem}
A set of points $P$ under the JS or $\chi^2$ divergence can be embedded into $\ell_2^{\bar{d}}$ with $\bar{d} = O \left(\frac{ n^2d^3}{\eps^2} \right)$ with $(1 + \eps)$ multiplicative error.  
\end{theorem}
For the both techniques, applying the Euclidean JL--Lemma can further reduce the dimension to $O \left(\frac{\log n}{\eps^2} \right)$ in the offline setting. 


\paragraph{Dimensionality reduction.}
\label{sec:struct-pres-dimens}
%
We then turn to the more challenging case of performing dimensionality reduction for information distances, where 
we wish to preserve not only the distances between pairs of points (distributions),
but also the underlying simplicial structure of the space, so that we can
continue to interpret coordinates in the new space as probabilities. This notion
of a \emph{structure-preserving} dimensionality reduction is implicit when
dealing with normed spaces (since we always map a normed space to another), but
requires an explicit mapping when dealing with more structured
spaces.  We prove an analog of the classical JL--Lemma : 

\begin{theorem}
For the Jenson-Shannon, Hellinger, and $\chi^2$ divergences, there exists a structure preserving dimensionality reduction from the high dimensional simplex $\Delta_d$ to a low dimensional simplex $\Delta_k$, where $k = O((\log n)/\eps^2)$. 
\end{theorem}

The theorem extends to ``well-behaved" $f$-divergences (See Section \ref{sec:background} for a precise definition). Moreover, the dimensionality reduction is constructive for any divergence with a  finite dimensional kernel (such as the Hellinger divergence), or an infinite dimensional Kernel that can be sketched in finite space, as we show is feasible for the JS and $\chi^2$ divergences.
 
\paragraph{Our techniques.}
\label{sec:dimens-reduct}

The unifying approach of our three results---sketching,  embedding into $\ell_2^2$, and  dimensionality reduction---is to analyze carefully the infinite dimensional kernel of the information divergences. Quantizing and truncating the kernel yields the sketching result,  sampling repeatedly from it produces an embedding into $\ell_2^2$. Finally given such an embedding, we show how to perform dimensionality reduction  by proving that each of the divergences admits a region of the simplex where it is similar to $\ell_2^2$. We point out that to the best of our knowledge, this is the first result that explicitly uses the kernel representation of these information distances to build approximate geometric structures; while the \emph{existence} of a kernel for the Jensen--Shannon distance was well-known, this structure had never been exploited for algorithmic advantage.

%% file: related.tex
\section{Related Work}\label{sec:related-work}

The works by~\citet{1365067}, and then by~\citet{vedaldi2012efficient} 
study embeddings of information divergences into an infinite dimensional Hilbert space by representing them as an integral along a one-dimensional curve in $\mathbb{C}$. Vedaldi and Zisserman give an explicit formulation of this kernel for JS and $\chi^2$ divergences, for which a discretization (by quantizing and truncating) yields an additive error embedding into a finite dimensional $\ell_2^2$. However, they do not obtain quantitative bounds on the dimension of target space needed or address the question of multiplicative approximation guarantees. 

In the realm of sketches,~\citet*{GIM} show $\Omega(n)$ space (where $n$ is the length of the stream) is required in the strict turnstile model even for a constant factor multiplicative approximation. These bounds hold for a wide range of information divergences, including JS, Hellinger and the $\chi^2$ divergences. They show however that an \emph{additive} error of $\eps$ can be achieved using $O \left(\frac{1}{\eps^3} \log n \right)$ space. 
In contrast, one can indeed achieve a multiplicative approximation in the aggregate streaming model for information divergences that have a finite dimensional embedding into $\ell_2^2$. For instance,~\citet{guhastreaming} observe that for the Hellinger distance that has a trivial such embedding, sketching is equivalent to sketching $\ell_2^2$ and hence may be done up to a $(1+\eps)$-multiplicative approximation in $\frac{1}{\eps^2} \log n$ space. This immediately implies a constant factor approximation of JS and $\chi^2$ divergences in the same space, but no bounds have been known prior to our work for a $(1+\eps)$-sketching result for JS and $\chi^2$ divergences in \emph{any} streaming model.

Moving onto dimensionality reduction from simplex to simplex, in the only other work we are aware of,~\citet*{kyng2010} show a limited dimensionality reduction result for the Hellinger distance. Their
approach works by showing that if the input points lie in a specific region of
the simplex, then a standard random projection will keep the points on a
lower-dimensional simplex while preserving the distances
approximately. Unfortunately, this region is a small ball centered in the interior of
the simplex, which further shrinks with the dimension. This is in sharp contrast to our work here, where the input points are unconstrained. 

While it does not admit a kernel, the $\ell_1$ distance is also an
$f$-divergence, and it is therefore natural to investigate its potential
connection with the measures we study here. For $\ell_1$, it is well known that
significant dimensionality reduction is not possible: an embedding with
distortion $1+\eps$ requires the points to be embedded in  $n^{1 - O \left( \log
    \frac{1}{\eps} \right)}$ dimensions, which is nearly linear. This result was
proved (and strengthened) in a series of
results~\citep{andoni2011near,regev,lee2004embedding,brinkman2005impossibility}. 

The general literature of sketching and embeddability in normed spaces is too
extensive to be reviewed here: we point the reader to \citet{2014sketching} for
a full discussion of results in this area. One of the most famous applications
of dimension reduction is the Johnson--Lindenstrauss(JL) Lemma, which states
that any set of $n$ points in $\ell_2^2$ can be embedded into
$O \left( \frac{\log n}{\eps^2} \right)$ dimensions in the same space while
preserving pairwise distances to within $(1 \pm \eps)$. This result has
become a core step in algorithms for near neighbor search~\citep{ailon2006, andoni2006}, speeding up clustering
algorithms~\citep{clusterdimred}, and efficient approximation of
matrices~\citep{matrixdimred}, among many others.

Although sketching, embeddability, and dimensionality reduction are related
operations, they are not always equivalent. For example, even though $\ell_1$ and
$\ell_2$ have very different behavior under dimensionality reduction, they can
both be sketched to an arbitrary error in the turnstile model (and in fact any
$\ell_p$ norm, $p \le 2$ can be sketched using $p$-stable
distributions~\citep{indyk2000stable}). In the offline 
setting,~\citet{2014sketching} show that sketching and embedding of normed spaces are equivalent: for any finite-dimensional normed space $X$, a constant distortion and space sketching algorithm 
for $X$ exists if and only if there exists a linear embedding of $X$ 
into $\ell_{1-  \eps}$. 

\section{Background}
\label{sec:background}
In this section, we define precisely the class of information divergences that we work with, and their specific properties that allow us to obtain sketching, embedding, and dimensionality results.
For what follows $\Delta_d$ denotes the \emph{$d$-simplex}: $\Delta_d = \{ (x_1, \ldots, x_d) \mid \sum x_i = 1 $ and $x_i \geq 0, \forall i \}$.  Let
$[d] = \{ 1, \ldots, d \}$. 
\begin{definition}[$f$-divergence]
  Let $p$ and $q$ be two distributions on $[n]$. A convex function  $f: [0,
  \infty) \to \reals$ such that $f(1) = 0$ gives rise to an 
  \emph{$f$-divergence}
  $D_f: \Delta_{d} \to \reals$ as:
  \[D_f(p,q) = \sum_{i=1}^d p_i \cdot f \left(\frac{q_i}{p_i} \right),\]
  where we define $0 \cdot f(0 /0) = 0$,  $a\cdot f(0/a) = a \cdot \lim_{u \to 0} f(u)$, and  $0 \cdot f(a/0) = a \cdot \lim_{u \to \infty} f(u)/u$.
\end{definition}

\begin{definition}[Regular distance]
  We call a distance function $D: X \to \reals$ \emph{regular} if there 
  exists a feature map $\phi: X \rightarrow V$, where $V$ is a (possibly infinite dimensional) Hilbert
  space, such that:
  \[ D(x,y) = \| \phi(x) - \phi(y) \|^2 \quad \forall x,y \in X. \]
\end{definition}

The work of~\cite{1365067} establishes that $\js$ is regular; 
\cite{vedaldi2012efficient} construct an explicit feature map for the JS kernel, as $\phi(x) = \int_{-\infty}^{+\infty} \Psi_x(\omega) \diff \omega$, where $\Psi_x(\omega): \reals \rightarrow \mathbb{C}$ is given by
\[ \Psi_x(\omega) = \exp(i \omega \ln x) \sqrt{\frac{2x \sech(\pi
    \omega)}{(\ln 4)(1 + 4\omega^2)}}.\] 

Hence we have for $x$, $y \in \reals$, $\js(x,y) = \| \phi(x) - \phi(y) \|^2 = \int_{-\infty}^{+\infty} \| \Psi_x(\omega) - \Psi_y(\omega) \|^2    \diff \omega$.
The ``embedding'' for a given distribution $p \in \Delta_{d}$ is then 
the concatenation of the functions $\phi(p_i)$, i.e.,
$\phi(p) = (\phi_{p_1}, \ldots, \phi_{p_d})$. 
\begin{definition}[Well-behaved divergence]
  A \emph{well-behaved} $f$-divergence is a regular $f$-divergence such that $f(1) = 0$, $f'(1) = 0$, $f''(1) > 0$,
  and $f'''(1)$ exists.
\end{definition}

In this paper, we will focus on the following well-behaved $f$-divergences.
\begin{definition}
The \emph{Jensen--Shannon} (JS), \emph{Hellinger}, and $\chi^2$
divergences between distributions $p$ and $q$ are defined as:
\begin{eqnarray*}
\js(p,q) & = & \sum_i p_i \log \frac{2p_i}{p_i + q_i} + q_i \log \frac{2q_i}{p_i+q_i}, \\
\he(p,q) & = & \sum_i (\sqrt{p_i} - \sqrt{q_i})^2, \\
\chi^2(p,q) & = & \sum_i \frac{ (p_i - q_i)^2 }{p_i + q_i}.
\end{eqnarray*}
\end{definition}


%% file: sketchjs.tex
\section{Embedding JS into $\ell_2^2$}\label{sec:embed-js}
We present two algorithms for embedding JS into $\ell_2^2$. The first is deterministic and gives an additive error approximation whereas the second is
randomized but yields a multiplicative 
approximation in an offline setting.  The advantage of the first algorithm is that it can be 
realized in the streaming model, and if we make a standard assumption of polynomial
precision in the streaming input, yields a $(1+\eps)$-multiplicative approximation as well in this setting.

We derive some terms in the kernel representation of $\js(x,y)$ which we will find convenient. First, the explicit formulation in Section \ref{sec:background} yields that for $x$, $y \in \reals$:
\begin{align*}
\js(x,y) &= \int_{-\infty}^{+\infty}\left \| e^{i \omega \ln x} \sqrt{\frac{2x \sech(\pi \omega)}{(\ln 4)(1 + 4\omega^2)}}  -
	      e^{i \omega \ln y} \sqrt{\frac{2y \sech(\pi \omega)}{(\ln 4)(1 + 4\omega^2)}} \right \|^2 \diff \omega \\
	&= \int_{-\infty}^{+\infty} \left(\frac{2 \sech(\pi \omega)}{(\ln 4)(1 + 4\omega^2)} \right) \| \sqrt{x} e^{i \omega \ln x}  -
	      \sqrt{y} e^{i \omega \ln y} \|^2 \diff \omega.
\end{align*}
 For convenience, we now define:
\begin{align*}
h(x,y,\w) 
& = \| \sqrt{x} e^{i \omega \ln x} - \sqrt{y} e^{i \omega \ln y} \|^2  \\
& =(\sqrt{x} \cos (\w \ln x)-\sqrt{y} \cos(\w \ln y))^2+ (\sqrt{x} \sin(\w \ln x )-\sqrt{y} \sin (\w \ln y))^2,
\end{align*} 
and
\[\kappa(\w) = \frac{2\sech(\pi \w)}{(\ln 4)(1+4\w^2)} \ .\]
 We can then write $\js(p,q)=\sum_{i=1}^d f_{J}(p_i,q_i)$ where $$f_{J}(x,y)=\int_{-\infty}^{\infty} h(x,y,\w) \kappa(\w) \diff \w 
= x \log \left(\frac{2x}{x+y} \right) + y \log \left(\frac{2y}{x+y} \right).$$
 It is easy to verify that $\kappa(\w)$ is a distribution, i.e.,
 $\int_{-\infty}^{\infty} \kappa(\w) d\w = 1$.

\subsection{Deterministic embedding}

We will produce an embedding $\phi(p) = (\phi_{p_1}, \ldots, \phi_{p_d})$, 
where each $\phi_{p_i}$ is an integral that we can discretize by 
quantizing and truncating carefully.

\begin{algorithm}
\caption{Embed $p \in \Delta_{d}$ under JS into $\ell_2^2$.}
\label{algo:embeda}
\DontPrintSemicolon 
\KwIn{$p=\{p_1, \ldots, p_d\}$ where coordinates are ordered by arrival. }
\KwOut{A vector $c^p$ of length $O \left( \frac{d^2}{\eps} \log \frac{d}{\eps} \ \right)$}
$\ell \gets 1$; $J \gets \roundup{\frac{32d}{\eps} \ln \left( \frac{8d}{\eps} \right)}$, 
\; \\
\For{$j \gets -J$ \textbf{to} $J$} {
	$w_j \gets j \times \eps/32d$ \;
}
\For{$i \gets 1$ \textbf{to} $d$} {
	\For{$j \gets -J $ \textbf{to}  $J - 1$} {	
  		$a^p_{\ell} \gets \sqrt{p_i} \cos (\w_j  \ln p_i) \sqrt{\int_{\w_j}^{\w_{j+1}} \kappa(\w) d\w}$ \;  \\
		$b^p_{\ell} \gets \sqrt{p_i} \sin (\w_j  \ln p_i) \sqrt{\int_{\w_j}^{\w_{j+1}} \kappa(\w) d\w}$ \;  \\
	 	$\ell \gets \ell +1$ \;  \\
	}
}
\Return{$a^p$ concatenated with $b^p$.}\;
\end{algorithm}

To analyze Algorithm \ref{algo:embeda}, we first obtain bounds on the function $h$ and its derivative.
\begin{lemma}
For $0 \leq x, y, \leq 1$, we have 
$0 \leq h(x,y,\w)\leq 2$
and
$\left |\frac{\partial h(x,y,\w)}{\partial \w} \right |\leq 16$.
\end{lemma}
\begin{proof}
Clearly $h(x, y, \w) \geq 0$.  Furthermore, since 
$0 \leq x, y \leq 1$, we have
\[
h(x,y,\w) 
\leq \left| \sqrt{x} e^{i \w \ln x} \right|^2
  + \left| \sqrt{y} e^{i \w \ln y} \right|^2
= x + y
\leq 2.
\]
\begin{eqnarray*}
& & \mbox{Next, } \left |\frac{\partial h(x,y,\w)}{\partial \w} \right | \\
&=& 
\left | 2  \left(\sqrt{x} \cos(\w \ln x)-\sqrt{y} \cos(\w \ln y )\right) \left(-\sqrt{x} \sin(\w \ln x ) \ln x + \sqrt{y} \sin(\w\ln y ) \ln y \right) \right .\\
& & \quad \left . + 2\left(\sqrt{x} \sin(\w \ln x)-\sqrt{y} \sin(\w \ln y)\right)\left(\sqrt{x} \cos(\w \ln x ) \ln x - \sqrt{y} \cos(\w \ln y) \ln y\right)  \right | \\
&\leq & 
\left | 2 \left(\sqrt{x} +\sqrt{y} \right) \left(\sqrt{x} \ln x + \sqrt{y} \ln y\right)  \right |
 +2 \left | \left(\sqrt{x} + \sqrt{y} \right) \left(\sqrt{x} \ln x + \sqrt{y} \ln y \right) \right | \enspace \leq \enspace 16,
\end{eqnarray*}
where the last inequality follows since
$\max_{0\leq x\leq 1} |\sqrt{x} \ln x| < 1$.
\end{proof}
The next two steps are useful to approximate the infinite-dimensional 
continuous representation by a finite-dimensional discrete representation 
by appropriately truncating and quantizing the integral.
\begin{lemma}[Truncation]\label{lem:truncate}
For $t\geq \ln (4/\eps)$,
\[f_{J}(x,y)\geq \int_{-t}^{t} h(x,y,\w) \kappa(\w) d\w
\geq f_{J}(x,y)-\eps \ .\]
\end{lemma}
\begin{proof}
The first inequality follows since $h(x,y,\w)\geq 0$. 
For the second inequality, we use $h(x,y,\w)\leq 2$:
\[
\int_{-\infty}^{-t} h(x,y,\w) \kappa(\w) d\w
+
\int_{t}^{\infty} h(x,y,\w) \kappa(\w) d\w
\leq 4\int_{t}^{\infty} \kappa(\w)  d\w
< 4\int_{t}^{\infty} \frac{4e^{-\pi \w} }{\ln 4 }  d\w
<4 e^{- t}\leq \eps
\]
where the last line follows if $t\geq \ln (4/\eps)$.
\end{proof}

Define $\w_i=\eps i/16 $ for $i\in \{\ldots, -2,-1,0,1,2,\ldots \}$ and  
$\tilde{h}(x,y,\w)=h(x,y,\w_i)  \mbox{ where } 
i=\max\{j \mid \omega_j\leq \omega\}$. 
%

\begin{lemma}[Quantization]\label{lem:quantize}
For any $a,b$,
\[\int_{a}^{b} h(x,y,\w) \kappa(\w) d\w
=\int_{a}^{b} \tilde{h}(x,y,\w) \kappa(\w) d\w \pm \eps \ .\]
\end{lemma}
\begin{proof}
First note that
\[
|\tilde{h}(x,y,\w)-h(x,y,\w)| \leq \left( \frac{\eps}{16} \right)
\cdot \max_{x,y\in [0,1],\omega} \left |\frac{\partial h(x,y,\w)}{\partial \w} \right |
\leq \eps \ .\]
Hence,
$
\left | \int_{-a}^{b} \tilde{h}(x,y,\w) \kappa(\w) d\w - \int_{-a}^{b} h (x,y,\w) \kappa(\w) d\w \right |
\leq 
\left | \int_{-a}^{b} \eps \kappa(\w) d\w \right |
 \leq \eps$.
\end{proof}

Given a real number $z$, define vectors $\vecv^z$ and 
$\vecu^z$ indexed by $i\in \{-i^*, \ldots, -2,-1,0,1,2,\ldots i^* \}$ where 
$i^* = \roundup{16 \eps^{-1} \ln (4/\eps)}$ by:
\[
\vecv^z = \sqrt{z} \cos (\w_i  \ln z) \sqrt{\int_{\w_i}^{\w_{i+1}} \kappa(\w) d\w},
\quad
\vecu^z = \sqrt{z} \sin (\w_i  \ln z) \sqrt{\int_{\w_i}^{\w_{i+1}} \kappa(\w) d\w}, \\
\]
and note that 
\begin{eqnarray*}
 (\vecv^x_i-\vecv^y_i)^2+(\vecu^x_i-\vecu^y_i)^2 
 =  h(x,y,\w_i)  \int_{\w_i}^{\w_{i+1}} \kappa(\w) d\w.
\end{eqnarray*}
Therefore,
\begin{eqnarray*}
\|\vecv^x-\vecv^y\|_2^2 + \|\vecu^x-\vecu^y\|_2^2 
&=& 
\int_{w_{-i^*}}^{w_{i^*+1}} \tilde{h}(x,y,\w) \kappa(\w) d\w 
\enspace = \enspace
\int_{w_{-i^*}}^{w_{i^*+1}} h(x,y,\w) \kappa(\w) d\w \pm \eps \\
&=& \int_{-\infty}^{\infty} h(x,y,\w) \kappa(\w) d\w \pm 2 \eps 
\enspace = \enspace f_{J}(x,y)\pm 2 \eps,
\end{eqnarray*}
where the second to last line follows from Lemma \ref{lem:quantize} and the last line follows from Lemma \ref{lem:truncate}, since  $\min (|w_{-i^*}|, w_{i^*+1})\geq \ln (4/\eps)$.

Define the vector $\veca^p$ to be the vector generated by concatenating $\vecv^{p_i}$ and $\vecu^{p_i}$ for $i\in [d]$. Then if follows that 
\[\|\veca^p-\veca^q\|_2^2 =\js(p,q) \pm 2\eps d\]
Hence we have reduced the problem of estimating $\js(p,q)$ to $\ell_2$ estimation. Rescaling $\eps \leftarrow  \eps/(2d)$ ensures the additive error is $\eps$ while the length of the vectors $\veca^p$ and $\veca^q$ is $O \left(\frac{d^2}{\eps}\log \frac{d}{\eps} \right)$. 
\begin{theorem}
\label{thm:jsmain}
Algorithm \ref{algo:embeda} embeds a set $P$ of points under JS 
into $O \left(\frac{d^2}{\eps}\log \frac{d}{\eps}  \right)$ dimensions 
under $\ell_2^2$ with $\eps$ additive error, independent of the size of $|P|$.
\end{theorem}
Note that using the JL-Lemma, the dimensionality of the target space can be
reduced to $O \left(\frac{\log |P|}{ \eps^2} \right)$.  Theorem~\ref{thm:jsmain},
along with the AMS sketch of~\cite{ams}, and the standard assumption of  polynomial precision 
immediately implies:
\begin{corollary}
There is an algorithm that works in the aggregate streaming model to 
approximate JS to within $(1+\eps)$-multiplicative factor using 
$O \left(\frac{1}{\eps^2} \log \frac{1}{ \eps} \log d \right)$ space.
\end{corollary}
As noted earlier, this is the first algorithm in the aggregate streaming 
model to obtain an $(1 + \eps)$-multiplicative approximation to JS,
which contrasts against linear space lower bounds for the same problem
in the update streaming model.

%% file: embedjs.tex
\subsection{Randomized embedding}\label{sec:multiplicative}

In this section we show how to embed $n$ points of $\js$ into $\ell_2^{\bar{d}}$ with $(1+\eps)$ distortion where $\bar{d}=O( n^2 d^3 \eps^{-2})$. \footnote{If we ignore precision constraints on sampling from a continuous distribution in a streaming algorithm, then this also would yield a sketching bound of $O(d^3 \eps^{-2})$ for a $(1+ \eps)$ multiplicative approximation.} 

For fixed $x,y,\in [0,1]$, we first consider the random variable $T$
where $T$ takes the value $h(x,y,\omega)$ with probability $\kappa(\omega)$. 
(Recall that $\kappa(\cdot)$ is a distribution.)  We compute the
first and second moments of $T$.
\begin{theorem}
$E[T]=f_{J}(x,y)$ and $\var[T]\leq 36 (f_{J}(x,y ))^2$.
\end{theorem}
\begin{proof}
The expectation follows immediately from the definition:
\[
E[T]=\int_{-\infty}^{\infty} h(x,y,\omega) \kappa(\omega) d\omega =f_{J}(x,y).
\]
To bound the variance it will be useful to define the function
$f_{H}(x,y)=(\sqrt{x}-\sqrt{y})^2$ corresponding to the one-dimensional Hellinger distance that is related to $f_{J}(x,y)$ as follows. We now state two claims regarding $f_{H}(x,y)$ and $f_{\chi}(x,y)$:

\begin{claim}\label{claim:fg}
For all $x, y \in [0,1]$, $f_{H}(x,y)\leq 2f_{J}(x,y) $.
\end{claim}
\begin{proof}
Let $f_{\chi}(x,y)=\frac{ (x-y)^2}{x+y}$ correspond to the one-dimensional $\chi^2$ distance.  Then, we have
\begin{align*}
\frac{f_{\chi}(x,y)}{f_{H}(x,y)} &
= \frac{ (x-y)^2}{(x+y)(\sqrt{x} - \sqrt{y})^2} 
= \frac{(\sqrt{x} + \sqrt{y})^2}{x+y} = \frac{x+y + 2 \sqrt{xy}}{x+y} \geq 1 \ .
\end{align*}
This shows that $f_{H}(x,y) \leq f_{\chi}(x,y)$.  To show $f_{\chi}(x,y)\leq 2f_{J}(x,y)$ we refer the reader to~\cite[Section 3]{topsoe2000}.  Combining these two relationships gives us our claim.  
\end{proof}

We then bound $h(x,y,\omega)$ in terms of $f_{H}(x,y)$ as follows.
\begin{claim}\label{claim:hg}
For all $x,y\in [0,1],\omega\in \reals$, 
$h(x,y,\omega)\leq f_{H}(x,y) (1+2 |\omega|)^2$.
\end{claim}
\begin{proof}
Without loss of generality, assume $x\geq y$.
\begin{eqnarray*}
\sqrt{h(x,y,\omega)}
&=& |\sqrt{x} \cdot e^{i \omega \ln x} - \sqrt{y} \cdot  e^{i \omega \ln y}| \\
& \leq & |\sqrt{x}  \cdot e^{i \omega \ln x} - \sqrt{y} \cdot e^{i \omega \ln x}|+|\sqrt{y} \cdot e^{i \omega \ln x} - \sqrt{y}  \cdot e^{i \omega \ln y}| \\
&=& |\sqrt{x}-\sqrt{y}|+\sqrt{y} \cdot 
  | e^{i \omega \ln x} - e^{i \omega \ln y} | \\
&=& |\sqrt{x}-\sqrt{y}|+\sqrt{y} \cdot 2 \cdot |\sin(\omega \ln(x/y)/2)| \\ 
&\leq & \sqrt{f_{H}(x,y)} + \sqrt{y} \cdot 2 \cdot  |\omega \ln(\sqrt{x/y})| \\
&\leq & \sqrt{f_{H}(x,y)} + \sqrt{y}\cdot 2 \cdot  |\sqrt{x/y}-1|
 \cdot | \omega|   \\
 &= &  \sqrt{f_{H}(x,y)} + 2 \sqrt{f_{H}(x,y)} \cdot | \omega|
 \end{eqnarray*}
 and hence $h(x,y,\omega)\leq f_{H}(x,y) (1 + 2 |\omega| )^2$ as required.
\end{proof}

These claims allow us to bound the variance:
\begin{eqnarray*}
\var[T]\leq E[T^2] = \int_{-\infty}^{\infty} (h(x,y,\w))^2 \kappa(\omega) d\w
&\leq & f_{H}(x,y)^2 \int_{-\infty}^{\infty} (1+2 |\omega|)^4  \kappa(\omega) d\w 
\\ 
& =&  f_{H}(x,y)^2 \cdot 8.94 \enspace < \enspace 36 f_{J}(x,y)^2,
\end{eqnarray*}
\end{proof}
This naturally gives rise to the following algorithm.
\begin{algorithm}
\caption{Embeds point $p \in \Delta_{d}$ under JS into $\ell_2^2$.}
\label{algo:embedb}
\DontPrintSemicolon 
\KwIn{$p=\{p_1, \ldots, p_d\}$. }
\KwOut{A vector $c^p$ of length $O \left( n^2 d^3 \eps^{-2}\right)$}
$\ell \gets 1$; $s \gets \roundup{36 n^2 d^2 \eps^{-2}}$ \; \\
\For{$j \gets 1 $ \textbf{to} $s$} {
	$\w_j \gets$  a draw from $\kappa(\w)$;
}
\For{$i \gets 1$ \textbf{to} $d$} {
	\For{$j \gets 1 $ \textbf{to} $s$} {	
  		$a^p_{\ell} \gets \left( \sqrt{p_i} \cos (\w_j \ln p_i)/s \right)$ \;  \\
		$b^p_{\ell} \gets \left( \sqrt{p_i} \sin (\w_j \ln p_i) /s \right)$ \;  \\
	 	$\ell \gets \ell +1$ \;  \\
	}
}
\Return{$a^p$ concatenated with $b^p$.}\;
\end{algorithm}
Let $\omega_1, \ldots, \omega_t$ be $t$ independent samples chosen according to $\kappa(\omega)$. For any distribution $p$ on $[d]$, define vectors $\vecv^p, \vecu^p \in \reals^{td}$ where, for $i\in [d], j\in [t]$,
\[\vecv^p_{i,j}= \sqrt{p_i} \cdot \cos ( \omega_j \ln p_i)/t,
\quad
\vecu^p_{i,j}= \sqrt{p_i} \cdot \sin ( \omega_j \ln p_i)/t. \]
Let $\vecv^p_{i}$ be a concatenation of $\vecv^p_{i,j}$ and 
$\vecu^p_{i,j}$ over all $j\in [t]$.
Then note that $E[ \|\vecv^p_{i}-\vecv^q_{i}\|^2_2] =f_{J}(p_i,q_i)$  and $\var[\|\vecv^p_{i}-\vecv^q_{i}\|^2_2] \leq 36(f_{J}(p_i,q_i))^2/t.$
Hence, for $t=36 n^2 d^2 \eps^{-2}$, by an application of the Chebyshev bound,
\begin{equation}\label{eq:cheby}
\Pr[| \|\vecv^p_{i}-\vecv^q_{i}\|^2_2-f_{J}(p_i,q_i)|\geq \eps f_{J}(x,y)]
\leq 36 \eps^{-2}/t= (nd)^{-2}.
\end{equation}
By an application of the union bound over all pairs of points:
\[
\Pr[\exists i\in [d]~,~  p,q \in P |\|\vecv^p_{i}-\vecv^q_{i}\|^2_2-f_{J}(p_i,q_i)|\geq \eps f_{J}(p_i,q_i)]
\leq 1/d.
\]
And hence, if $\vecv^p$ is a concatenation of $\vecv^p_i$ over all 
$i\in [d]$, then with probability at least $1-1/d$ it holds for all $p$,$q \in P$: 
\[
(1-\eps ) \js(p,q)\leq \|\vecv^p-\vecv^q\|\leq (1+\eps ) \js(p,q). 
\]
The final length of the vectors is then $td = 36 n^2 d^3 \eps^{-2}$ for approximately preserving distances between every pair of points with probability at least $1 - \frac{1}{d}$.  This can be reduced further to $O(\log n / \eps^2)$ by simply applying the JL-Lemma.


%% file: embedtr.tex
\section{Embedding $\chi^2$ into $\ell_2^2$}

We give here two algorithms for embedding the $\chi^2$ divergence into $\ell_2^2$.  The computation and resulting two algorithms are highly analogous to Section \ref{sec:embed-js}. First, the explicit formulation given by~\citet{vedaldi2012efficient} yields that for $x$, $y \in \reals$:
\begin{align*}
\chi^2(x,y) &= \int_{-\infty}^{+\infty}\left \| e^{i \omega \ln x} \sqrt{x \sech(\pi \omega)}  -
	      e^{i \omega \ln y} \sqrt{y \sech(\pi \omega)} \right \|^2 \diff \omega \\
	&= \int_{-\infty}^{+\infty} \left( \sech(\pi \omega) \right) \| \sqrt{x} e^{i \omega \ln x}  -
	      \sqrt{y} e^{i \omega \ln y} \|^2 \diff \omega.
\end{align*}
 For convenience, we now define:
\[h(x,y,\w)  = \| \sqrt{x} e^{i \omega \ln x} - \sqrt{y} e^{i \omega \ln y} \|^2  \, \] 
and $\kappa_{\chi}(\w) = \sech(\pi \w) .$

 We can then write $\chi^2(p,q)=\sum_{i=1}^d f_{\chi}(p_i,q_i)$ where $$f_{\chi}(x,y)=\int_{-\infty}^{\infty} h(x,y,\w) \kappa_{\chi}(\w) \diff \w 
= \frac{\left( x-y \right)^2}{x+y}.$$
 It is easy to verify that $\kappa_{\chi}(\w)$ is a distribution, i.e.,
 $\int_{-\infty}^{\infty} \kappa_{\chi}(\w) d\w = 1$.

\subsection{Deterministic embedding}
We will produce an embedding $\phi(p) = (\phi_{p_1}, \ldots, \phi_{p_d})$, 
where each $\phi_{p_i}$ is an integral that we discretize appropriately. 

\begin{algorithm}
\caption{Embed $p \in \Delta_{d}$ under $\chi^2$ into $\ell_2^2$.}
\label{algo:embedatr}
\DontPrintSemicolon 
\KwIn{$p=\{p_1, \ldots, p_d\}$ where coordinates are ordered by arrival. }
\KwOut{A vector $c^p$ of length $O \left( \frac{d^2}{\eps} \log \frac{d}{\eps} \ \right)$}
$\ell \gets 1$; $J \gets \roundup{\frac{32d}{\eps} \ln \left( \frac{6d}{\eps} \right)}$, 
\; \\
\For{$j \gets -J$ \textbf{to} $J$} {
	$w_j \gets j \times \eps/32d$ \;
}
\For{$i \gets 1$ \textbf{to} $d$} {
	\For{$j \gets -J $ \textbf{to}  $J - 1$} {	
  		$a^p_{\ell} \gets \sqrt{p_i} \cos (\w_j  \ln p_i) \sqrt{\int_{\w_j}^{\w_{j+1}} \kappa_{\chi}(\w) d\w}$ \;  \\
		$b^p_{\ell} \gets \sqrt{p_i} \sin (\w_j  \ln p_i) \sqrt{\int_{\w_j}^{\w_{j+1}} \kappa_{\chi}(\w) d\w}$ \;  \\
	 	$\ell \gets \ell +1$ \;  \\
	}
}
\Return{$a^p$ concatenated with $b^p$.}\;
\end{algorithm}

\begin{lemma}
For $0 \leq x, y, \leq 1$, we have 
$0 \leq h(x,y,\w)\leq 2$
and
$\left |\frac{\partial h(x,y,\w)}{\partial \w} \right |\leq 16$.
\end{lemma}

Similar to Section \ref{sec:embed-js}, the next two steps analyze truncating and quantizing the integral.
\begin{lemma}[Truncation]\label{lem:truncatetr}
For $t\geq \ln (3/\eps)$,
\[f_{\chi}(x,y)\geq \int_{-t}^{t} h(x,y,\w) \kappa_{\chi}(\w) d\w
\geq f_{\chi}(x,y)-\eps \ .\]
\end{lemma}
\begin{proof}
The first inequality follows since $h(x,y,\w)\geq 0$. 
For the second inequality, we use $h(x,y,\w)\leq 2$:
\[
\int_{-\infty}^{-t} h(x,y,\w) \kappa_{\chi}(\w) d\w
+
\int_{t}^{\infty} h(x,y,\w) \kappa_{\chi}(\w) d\w
\leq 4\int_{t}^{\infty} \kappa_{\chi}(\w)  d\w
< 4\int_{t}^{\infty} 2e^{-\pi \w}  d\w  
<3 e^{- t}\leq \eps
\]
where the last line follows if $t\geq \ln (3/\eps)$.
\end{proof}

Define $\w_i=\eps i/16 $ for $i\in \{\ldots, -2,-1,0,1,2,\ldots \}$ and  
$\tilde{h}(x,y,\w)=h(x,y,\w_i)  \mbox{ where } 
i=\max\{j \mid \omega_j\leq \omega\}$. We recall the following Lemma from Section \ref{sec:embed-js}:

\begin{lemma}[Quantization]\label{lem:quantizetr}
For any $a,b$,
\[\int_{a}^{b} h(x,y,\w) \kappa_{\chi}(\w) d\w
=\int_{a}^{b} \tilde{h}(x,y,\w) \kappa_{\chi}(\w) d\w \pm \eps \ .\]
\end{lemma}

Given a real number $z$, define vectors $\vecv^z$ and 
$\vecu^z$ indexed by $i\in \{-i^*, \ldots, -2,-1,0,1,2,\ldots i^* \}$ where 
$i^* = \roundup{16 \eps^{-1} \ln (3/\eps)}$ by:
\[
\vecv^z = \sqrt{z} \cos (\w_i  \ln z) \sqrt{\int_{\w_i}^{\w_{i+1}} \kappa_{\chi}(\w) d\w},
\quad
\vecu^z = \sqrt{z} \sin (\w_i  \ln z) \sqrt{\int_{\w_i}^{\w_{i+1}} \kappa_{\chi}(\w) d\w}, \\
\]
and note that 
\begin{eqnarray*}
 (\vecv^x_i-\vecv^y_i)^2+(\vecu^x_i-\vecu^y_i)^2 
 =  h(x,y,\w_i)  \int_{\w_i}^{\w_{i+1}} \kappa_{\chi}(\w) d\w.
\end{eqnarray*}
Therefore,
\begin{eqnarray*}
\|\vecv^x-\vecv^y\|_2^2 + \|\vecu^x-\vecu^y\|_2^2 
&=& 
\int_{w_{-i^*}}^{w_{i^*+1}} \tilde{h}(x,y,\w) \kappa_{\chi}(\w) d\w 
\enspace = \enspace
\int_{w_{-i^*}}^{w_{i^*+1}} h(x,y,\w) \kappa_{\chi}(\w) d\w \pm \eps \\
&=& \int_{-\infty}^{\infty} h(x,y,\w) \kappa_{\chi}(\w) d\w \pm 2 \eps 
\enspace = \enspace f_{\chi}(x,y)\pm 2 \eps,
\end{eqnarray*}
where the second to last line follows from Lemma \ref{lem:quantizetr} and the last line follows from Lemma \ref{lem:truncatetr}, since  $\min (|w_{-i^*}|, w_{i^*+1})\geq \ln (3/\eps)$.

Define the vector $\veca^p$ to be the vector generated by concatenating $\vecv^{p_i}$ and $\vecu^{p_i}$ for $i\in [d]$. Then if follows that 
\[\|\veca^p-\veca^q\|_2^2 =\chi^2(p,q) \pm 2\eps d\]
Hence we have reduced the problem of estimating $\chi^2(p,q)$ to $\ell_2$ estimation. Rescaling $\eps \leftarrow  \eps/(2d)$ ensures the additive error is $\eps$ while the length of the vectors $\veca^p$ and $\veca^q$ is $O \left(\frac{d^2}{\eps}\log \frac{d}{\eps} \right)$. 
\begin{theorem}
\label{thm:trmain}
Algorithm \ref{algo:embedatr} embeds a set $P$ of points under $\chi^2$ 
into $O \left(\frac{d^2}{\eps}\log \frac{d}{\eps}  \right)$ dimensions 
under $\ell_2^2$ with $\eps$ additive error, independent of the size of $|P|$.
\end{theorem}
  Theorem~\ref{thm:trmain}, along with the AMS sketch of~\cite{ams}, and the standard assumption of  polynomial precision 
immediately implies:
\begin{corollary}
There is an algorithm that works in the aggregate streaming model to 
approximate $\chi^2$ to within $(1+\eps)$-multiplicative factor using 
$O \left(\frac{1}{\eps^2} \log \frac{1}{ \eps} \log d \right)$ space.
\end{corollary}

%% file: embedtrsecond.tex
\subsection{Randomized embedding}\label{sec:multiplicativetr}

In this section we show how to embed $n$ points of $\chi^2$ into $\ell_2^{\bar{d}}$ with $(1+\eps)$ distortion where $\bar{d}=O( n^2 d^3 \eps^{-2})$. \footnote{If we ignore precision constraints on sampling from a continuous distribution in a streaming algorithm, then this also would yield a sketching bound of $O(d^3 \eps^{-2})$ for a $(1+ \eps)$ multiplicative approximation.} 

For fixed $x,y,\in [0,1]$, we first consider the random variable $T$
where $T$ takes the value $h(x,y,\omega)$ with probability $\kappa_{\chi}(\omega)$. 
(Recall that $\kappa_{\chi}(\cdot)$ is a distribution.)  We compute the
first and second moments of $T$.
\begin{theorem}
$E[T]=f_{\chi}(x,y)$ and $\var[T]\leq 23 (f_{\chi}(x,y ))^2$.
\end{theorem}
\begin{proof}
The expectation follows immediately from the definition:
\[
E[T]=\int_{-\infty}^{\infty} h(x,y,\omega) \kappa_{\chi}(\omega) d\omega =f_{\chi}(x,y).
\]
To bound the variance we will again use the function $f_{H}(x,y)=(\sqrt{x}-\sqrt{y})^2$ corresponding to the one-dimensional Hellinger distance. We now state two claims relating $f_{H}(x,y)$ and $f_{\chi}(x,y)$:

\begin{claim}\label{claim:fgtr}
For all $x, y \in [0,1]$, $f_{H}(x,y)\leq f_{\chi}(x,y) $.
\end{claim}
\begin{proof}
Let $f_{\chi}(x,y)=\frac{ (x-y)^2}{x+y}$ correspond to the one-dimensional $\chi^2$ distance.  Then, we have
\begin{align*}
\frac{f_{\chi}(x,y)}{f_{H}(x,y)} &
= \frac{ (x-y)^2}{(x+y)(\sqrt{x} - \sqrt{y})^2} 
= \frac{(\sqrt{x} + \sqrt{y})^2}{x+y} = \frac{x+y + 2 \sqrt{xy}}{x+y} \geq 1 \ .
\end{align*}
This shows that $f_{H}(x,y) \leq f_{\chi}(x,y)$.
\end{proof}

We then recall Claim \ref{claim:hg} bounding $h(x,y,\omega)$ in terms of $f_{H}(x,y)$ as follows.
\begin{claim}\label{claim:hgtr}
For all $x,y\in [0,1],\omega\in \reals$, 
$h(x,y,\omega)\leq f_{H}(x,y) (1+2 |\omega|)^2$.
\end{claim}

These claims allow us to bound the variance:
\begin{eqnarray*}
\var[T]\leq E[T^2] = \int_{-\infty}^{\infty} (h(x,y,\w))^2 \kappa_{\chi}(\omega) d\w
&\leq & f_{H}(x,y)^2 \int_{-\infty}^{\infty} (1+2 |\omega|)^4  \kappa_{\chi}(\omega) d\w 
\\ 
& =&  f_{H}(x,y)^2 \cdot 22.77 \enspace < \enspace 23 f_{\chi}(x,y)^2,
\end{eqnarray*}
\end{proof}
This naturally gives rise to the following algorithm.
\begin{algorithm}
\caption{Embeds point $p \in \Delta_{d}$ under $\chi^2$ into $\ell_2^2$.}
\label{algo:embedbtr}
\DontPrintSemicolon 
\KwIn{$p=\{p_1, \ldots, p_d\}$. }
\KwOut{A vector $a^p$ of length $O \left( n^2 d^3 \eps^{-2}\right)$}
$\ell \gets 1$; $s \gets \roundup{23 n^2 d^2 \eps^{-2}}$ \; \\
\For{$j \gets 1 $ \textbf{to} $s$} {
	$\w_j \gets$  a draw from $\kappa_{\chi}(\w)$;
}
\For{$i \gets 1$ \textbf{to} $d$} {
	\For{$j \gets 1 $ \textbf{to} $s$} {	
  		$a^p_{\ell} \gets \left( \sqrt{p_i} \cos (\w_j \ln p_i)/s \right)$ \;  \\
		$b^p_{\ell} \gets \left( \sqrt{p_i} \sin (\w_j \ln p_i) /s \right)$ \;  \\
	 	$\ell \gets \ell +1$ \;  \\
	}
}
\Return{$a^p$ concatenated with $b^p$.}\;
\end{algorithm}
Let $\omega_1, \ldots, \omega_t$ be $t$ independent samples chosen according to $\kappa_{\chi}(\omega)$. For any distribution $p$ on $[d]$, define vectors $\vecv^p, \vecu^p \in \reals^{td}$ where, for $i\in [d], j\in [t]$,
\[\vecv^p_{i,j}= \sqrt{p_i} \cdot \cos ( \omega_j \ln p_i)/t,
\quad
\vecu^p_{i,j}= \sqrt{p_i} \cdot \sin ( \omega_j \ln p_i)/t. \]
Let $\vecv^p_{i}$ be a concatenation of $\vecv^p_{i,j}$ and 
$\vecu^p_{i,j}$ over all $j\in [t]$.
Then note that $E[ \|\vecv^p_{i}-\vecv^q_{i}\|^2_2] =f_{\chi}(p_i,q_i)$  and $\var[\|\vecv^p_{i}-\vecv^q_{i}\|^2_2] \leq 23(f_{\chi}(p_i,q_i))^2/t.$
Hence, for $t=23 n^2 d^2 \eps^{-2}$, by an application of the Chebyshev bound,
\begin{equation}\label{eq:chebytr}
\Pr[| \|\vecv^p_{i}-\vecv^q_{i}\|^2_2-f_{\chi}(p_i,q_i)|\geq \eps f_{\chi}(x,y)]
\leq 23 \eps^{-2}/t= (nd)^{-2}.
\end{equation}
By an application of the union bound over all pairs of points:
\[
\Pr[\exists i\in [d]~,~  p,q \in P |\|\vecv^p_{i}-\vecv^q_{i}\|^2_2-f_{\chi}(p_i,q_i)|\geq \eps f_{\chi}(p_i,q_i)]
\leq 1/d.
\]
And hence, if $\vecv^p$ is a concatenation of $\vecv^p_i$ over all 
$i\in [d]$, then with probability at least $1-1/d$, 
\[
(1-\eps ) \chi^2(p,q)\leq \|\vecv^p-\vecv^q\|\leq (1+\eps ) \chi^2(p,q). 
\]
The final length of the vectors is then $td = 23 n^2 d^3 \eps^{-2}$ for approximately preserving distances between every pair of points with probability at least $1 - \frac{1}{d}$.  This can be reduced further to $O(\log n / \eps^2)$ by simply applying the JL-Lemma.

%% file: simptosimp.tex
\section{Dimensionality reduction}
The JL-Lemma has been instrumental for improving the speed and approximation ratios of learning algorithms.  In this section, we give a proof of the JL-analogue for well behaved $f$-divergences. Specifically, we show that  a set of $n$ points lying on a high dimensional simplex can be embedded to a $k= O(\log n /\eps^2)$-dimensional simplex, while approximately preserving the information distances between all pairs of points. This dimension reduction amounts to reducing the support of the distribution from $d$ to $k$, while approximately maintaining the divergences. 

Our proof uses $\ell_2^2$ as an intermediate space. On a high level, we first embed the points into a high (but finite) dimensional $\ell_2^2$ space, using the techniques we developed in Section \ref{sec:multiplicative}. We then use the Euclidean JL-Lemma to reduce the dimensionality, and remap the points into the interior of a simplex. Finally, we show that far away from the simplex boundaries, the well behaved $f$-divergences have the same structure as $\ell_2^2$, hence the embedding back into information spaces can be done with a simple translation and rescaling. Note that for $f$-divergences that have an embedding into finite dimensional $\ell_2^2$, the proof is constructive. 


\begin{algorithm}
\caption{\sc{Dimension Reduction for $D_f$}}
\label{algo:embedc}
\DontPrintSemicolon 
\KwIn{Set $P=\{p_1, \ldots, p_n\}$ of points on $\Delta_{d}$, error parameter $\eps$, constant $c_0(\eps, f)$ }
\KwOut{A set $\bar{P}$ of points on $\Delta_{k}$ where  $k=O \left(  \frac{\log n}{\eps^2} \right)$}
\begin{enumerate}\itemsep=0in
\item Embed $P$ into $\ell_2^2$ to obtain $P_1$ with error parameter $\nicefrac{\eps}{4}$. 
\item Apply Euclidean JL--Lemma with error $\frac{\eps}{4}$ to obtain $P_2$ in dimension $k = O \left( \frac{\log n}{\eps^2} \right)$ 
\item Remap $P_2$ to the plane $L=\{ x \in \reals^{k+1}  \mid \sum_i x_i =0 \}$  to obtain $P_3$ 
\item Scale $P_3$ to a ball of radius $c_0 \cdot \frac{\eps}{k+1}$ and center at the centroid of $\Delta_{k+1}$ to obtain $\bar{P}$.
\end{enumerate}
\end{algorithm}

To analyze the above algorithm, we recall the JL--Lemma~\citep{jl, lugosi}:
\begin{lemma}\label{lem:jl}
For any set of points $P$ in a (possibly infinite dimensional) Hilbert space $H$, there exists a randomized map $f \colon H \to \reals^k$, $k = O(\frac{\log n}{\eps^2})$ such that $\forall p, q \in P$ 
\begin{equation}
(1-\eps) \|p-q \|_2^2 \leq \|f(p) - f(q) \|_2^2 \leq (1+\eps) \|p-q \|_2^2,
\end{equation}
with high probability. 
\end{lemma}
We show the following simple corollary:
\begin{corollary}\label{cor:jl}
For any set of points $P$ in $H$ there exists a constant $t$ and a randomized map $f \colon H \to \Delta_{k+1}$, $k = O(\frac{\log n}{\eps^2})$  such that $\forall p, q \in P$:
\begin{equation}
(1-\eps) \|p-q \|_2^2 \leq t \|f(p) - f(q) \|_2^2 \leq (1+\eps) \|p-q \|_2^2,
\end{equation}
Furthermore for any small enough constant $r$, we may bound the domain of $f$ to be a ball $B$ of radius $r$ centered at the simplex centroid, $(\nicefrac{1}{k+1}, \nicefrac{1}{k+1}, \ldots, \nicefrac{1}{k+1})$ . 
\end{corollary}
\begin{proof}
Consider first the map of Lemma~\ref{lem:jl} from $\reals^d \to \reals^k$ . Now note that any set of points in $R^k$ can be isometrically embedded into the hyperplane $L=\{ x \in \reals^{k+1}  \mid \sum_i x_i =0 \}$. This follows by remapping the basis vectors of  $\reals^k$ to those of $L$. Finally since $L$ is parallel to the simplex plane, the entire pointset may be scaled by some factor
$t$ and then translated to fit in $\Delta_{k+1}$, or indeed in any ball of radius $r$ centered at the simplex centroid.
\end{proof}
We now show that any well-behaved $f$ divergence is nearly Euclidean near the simplex centroid.
\begin{lemma}\label{lem:identity}
Consider any well-behaved $f$ divergence $D_f$, and let $B_r$ be a ball of radius $r$ such that $B_r \subset \Delta_k$ and $B_r$ is centered at the simplex centroid. Then for any fixed $0<\eps<1$, there exists a choice of $r$ and scaling factor $t$ (both dependent on $k$) such that $\forall p,q \in B$:
\begin{equation}
(1- \eps) \|p-q \|_2^2 \leq t D_f(p,q) \leq (1+ \eps) \|p-q \|_2^2.
\end{equation} 
\end{lemma}

\begin{proof}
We consider arbitrary $p$, $q \in B_r$ and note that the assumptions imply each coordinate lies in the interval $I =[\frac{1}{k} -r, \frac{1}{k}+r]$.  Let $rk = \eps'$, then $I = [\frac{1 -\eps'}{k} , \frac{1+\eps'}{k}]$. We now prove the lemma for $p,q \in I$, the main result follows by considering $D_f$ and $\|\cdot\|^2$ coordinate by coordinate. 

By the definition of well-behaved $f$-divergences and Taylor's theorem, there exists a neighborhood $N$ of $1$, and function $\phi$ with $\lim_{x \to 1}\phi(1) = 0$  such that for all $x \in N$:
\begin{equation}\label{eq:secondapprox}
f(x) = f(1) + (x-1)f'(1) + \frac{(x-1)^2}{2} f''(1) + (x-1)^3 \phi(x)
= \frac{(x-1)^2}{2} f''(1) + (x-1)^3 \phi(x).
\end{equation}
Therefore:
\begin{align*}
\frac{D_f(p,q)}{\|p-q \|_2^2} 
&= \frac{p \cdot f\left(  \frac{q}{p}  \right)}{(p-q)^2}  \\ 
&=   \frac{p  \left(
\left(\frac{q - p}{p}\right)^2 \frac{f''(1)}{2} 
+ \left(\frac{q - p}{p} \right)^3 \phi\left( \frac{q}{p} \right) 
\right)}{(q - p)^2} \\
&=     \frac{f''(1)}{2p}   + \frac{q - p}{p^2} \phi\left(\frac{q}{p} \right). 
\end{align*}
Recall again that $p \in [\frac{1 -\eps'}{k} , \frac{1+\eps'}{k}]$ so the first term converges to the constant $2k f''(1)$ as $r$ grows smaller (and hence $\eps'$ decreases). Note also that the second term goes to $0$ with $r$, i.e., given  
a suitably small choice of $r$ we can make the term smaller than any desired constant. Hence, for every dimension $k$ and  $0 <\eps < 1$, there exists a radius of convergence $r$ such that for all $p,q \in B_r$:
\begin{equation}
(1- \eps) \|p-q \|_2^2  \leq \frac{1}{2k f''(1)} D_f(p,q) \leq (1+\eps)\|p-q \|_2^2. 
\end{equation}
\end{proof}
We note that the required value of $r$ can be computed easily for the Hellinger and $\chi^2$ divergence, and that $r$ behaves as 
$\frac{1}{k} \cdot c$ where $c = c(f, \epsilon)$ is a constant depending only on $\eps$ and the function $f$ and not on $k$ or $n$ .

To conclude the proof note that the overall distortion is bounded by the combination of errors due to the initial embedding into $P_1$, the application of JL-Lemma, and the final reinterpretation of the points in $\Delta_{k+1}$. The overall error is thus bounded by,  $(1 + \nicefrac{\eps}{4})^3 \leq 1+\eps$. 
%
%

\begin{theorem}
Consider a set $P \in \Delta_d$ of $n$ points under a well-behaved $f$-divergence $D_f$. Then there exists a $(1+ \eps)$ distortion embedding of $P$ into $\Delta_{k}$ under $D_f$ for some choice of $k$ bounded as $ O \left( \frac{\log n}{\eps^2} \right)$. Furthermore this embedding can be explicitly computed even for a well-behaved $f$-divergence with an infinite dimensional kernel, if the kernel can be approximated in finite dimensions within a multiplicative error as we show for JS and $\chi^2$.
\end{theorem}